\documentclass{cccg26}
\usepackage{graphicx,amssymb,amsmath}
\usepackage{makecell}
\usepackage{listings}
\usepackage[flushleft]{threeparttable}
\usepackage{xcolor}
\usepackage{hyperref}
{\makeatletter \hypersetup{pdftitle={\@title}}}

{\makeatletter
 \gdef\xxxmark{%
   \expandafter\ifx\csname @mpargs\endcsname\relax 
     \expandafter\ifx\csname @captype\endcsname\relax 
       \marginpar{xxx}
     \else
       xxx 
     \fi
   \else
     xxx 
   \fi}
 \gdef\xxx{\@ifnextchar[\xxx@lab\xxx@nolab}
 \long\gdef\xxx@lab[#1]#2{\textbf{[\xxxmark #2 ---{\sc #1}]}}
 \long\gdef\xxx@nolab#1{\textbf{[\xxxmark #1]}}
}

{\makeatletter \gdef\fps@figure{!htbp}}

\let\epsilon=\varepsilon
\def\defn#1{\textbf{\textit{\boldmath #1}}}

\title{Polyomino Nets Covering Three Different Boxes of Area 106
and Related Results}

\author{Erik D. Demaine \thanks{Massachusetts Institute of Technology, USA \texttt{edemaine@mit.edu}}
	\and
	Jenny Diomidova \thanks{Massachusetts Institute of Technology, USA \texttt{diomidova@mit.edu}}
    \and
    Nicole Jacobus \thanks{Stand-up Maths, UK \texttt{nicole@standupmaths.com}}
    \and
    Landon Kryger \thanks{Alum of Washington State University, USA \texttt{landon.kryger@gmail.com}}
    \and
    Matthew T. Parker \thanks{Stand-up Maths, UK \texttt{matt@standupmaths.com}}
    \and
    Michael Tardibuono \thanks{Alum of the University of Waterloo, Canada \texttt{michael.tardibuono.math@gmail.com}}
    \and
    Ryuhei Uehara \thanks{Japan Advanced Institute of Science and Technology, Japan \texttt{uehara@jaist.ac.jp}}
    \and
    Hanyu Alice Zhang \thanks{Cornell University, USA, \texttt{hz496@cornell.edu}}}

\index{Author, First}
\index{Researcher, Second}

\begin{document}
\thispagestyle{empty}
\maketitle
\begin{abstract}

We present new results for polyomino nets that fold into 2 and 3 different cuboids through a computer search. The main result is the finding of 40 nets that fold into all three different cuboids with a surface area of 106. The secondary results are the finding of infinite families of nets that fold into three cuboid shapes, and the calculation of the number of common nets between smaller cuboids. The algorithms used to make the searches feasible will also be explained. The algorithms include taking advantage of some hidden structures in the nets that fold into the $N\times 1 \times 1$ cuboids, taking advantage of how a lot of nets fold into cuboid shapes in a 'striped' way, and using a variant of Redelmeier's algorithm. The paper ends with open questions that encourage the reader to broaden our collective understanding of the subject of creating polyomino nets.
\end{abstract}

\section{Introduction}

In the 1990s, there was a curious discovery that the same development (\emph{i.e.} net) could fold and create multiple different shapes~\cite[p. 424]{notes_1999,geometric_folding_algo}. One line of inquiry into this topic is focused around whether or not a polyomino net (\emph{i.e.} squares glued together edge-wise) can fold into multiple different cuboids (\emph{i.e.} boxes). This area of study, the common developments of boxes, offers results that are visually interesting even for a lay audience, can be made into puzzles, and motivates the creation of algorithms that may be engineered into solutions for other applications. Previous studies have shown nets that cover three different cuboids~\cite{magic_lid}, as well as a proof that there are no nets that cover three cuboids of area 58 or less despite there being many nets that cover two cuboids~\cite{qian2025unfoldingboxeslocalconstraints}. For a list of papers studying this topic, see~\cite{Abel2011,mitani2009polygons,qian2025unfoldingboxeslocalconstraints,magic_lid, multiple_ways, intro_comp_origami,XU1210214, Xu-Uehara-area-30}.

In this paper, we will present new nets that cover three cuboids with areas smaller than previously known, share results for many different types of nets, and share ideas that may help us create algorithms that are better at finding nets covering multiple cuboids. Results in the paper follow the \defn{grid unfolding, no overlap} model of folding, which requires all nets be polyominoes that:

\begin{itemize}
    \item can only fold along the grid as defined by the polyomino, and
    \item can be laid on a flat surface without overlap.
\end{itemize}

Reference material and the code used to get these results were posted on zenodo.org~\cite{michael_tardibuono_2026_NetSolverClean, michael_tardibuono_2026_stack_net, michael_tardibuono_2026_solver_thorough, michael_tardibuono_2026_material}. See Appendix \ref{zenodo_appendix} for more details.

\section{The Non-Exhaustive Search Approach}

\begin{figure}
    \centering
    \includegraphics[width=0.5\linewidth]{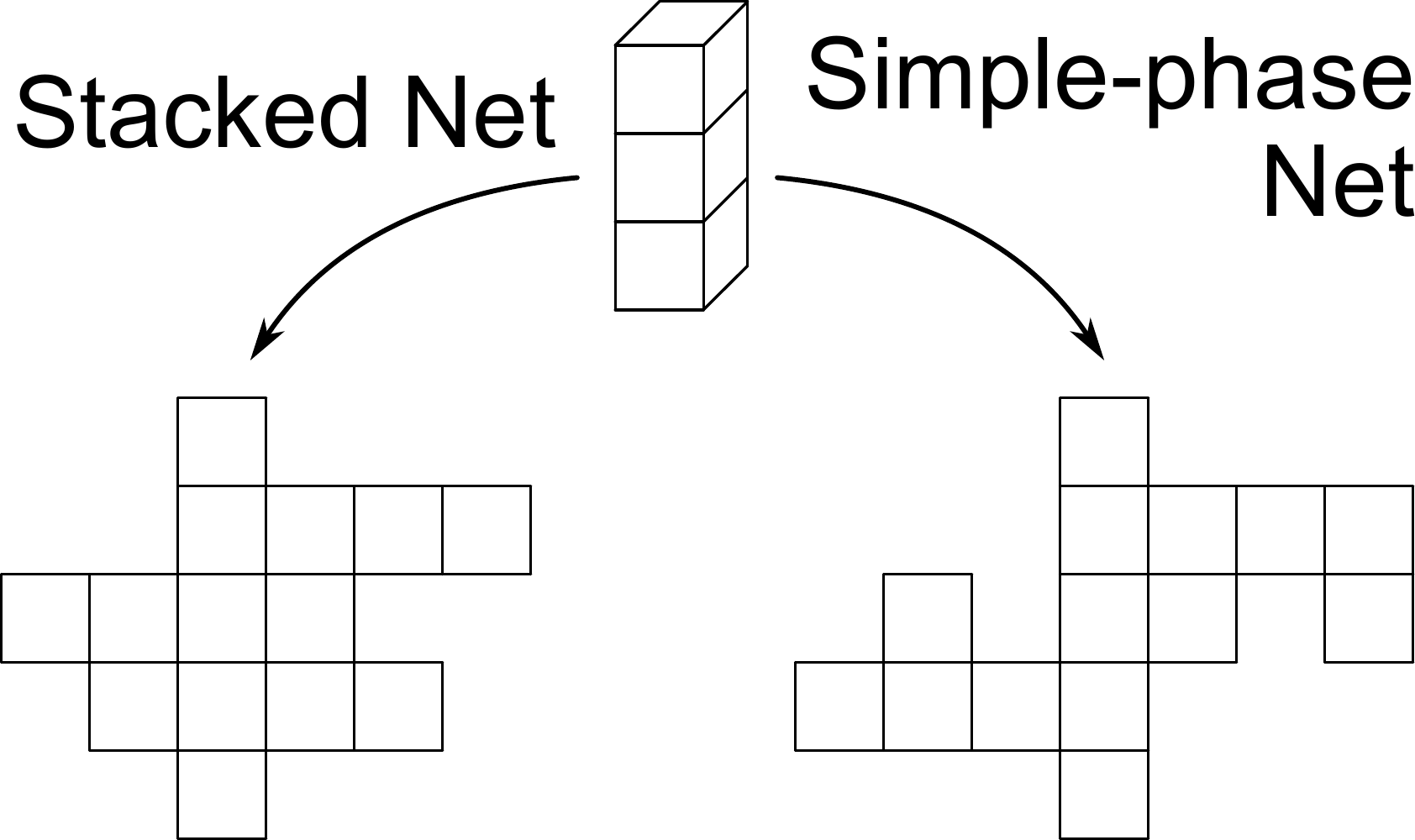}
    \caption{Examples of a \defn{stacked net} (left) and a \defn{simple-phase net} (right) of a $1\times 1\times 3$ cuboid.}
    \label{fig:stacked-simple-phase-examples}
\end{figure}

Although a thorough search of all nets for cuboids with a surface area of 46 didn't reveal any positive results for a net that can fold into 3 cuboids (of size $1\times 1\times 11$, $1\times 2\times 7$, and $1\times 3\times 5)$, we found that the $N\times 1 \times 1$ nets with the top and bottom $1\times 1$ sides only having one neighbour was a promising place to search because the nets looked simple and amenable to optimizations. We decided to search those nets in two different ways: 

\begin{enumerate}
    \item Nets where every layer that is not the top or bottom $1\times 1$ side has all 4 tiles attached to each other. We call these \defn{stacked nets}. The left hand side of Figure~\ref{fig:stacked-simple-phase-examples} shows a \defn{stacked net} for a $1\times 1\times 3$ cuboid.
    \item Nets that obey the rule that the top and bottom $1\times 1$ sides have only one neighbor. We call these \defn{simple-phase nets}. The right hand side of Figure~\ref{fig:stacked-simple-phase-examples} shows a \defn{simple-phase net} for a $1\times 1\times 3$ cuboid.
\end{enumerate}

\section{Results after Exhaustively Searching for Stacked Nets Covering the $N \times 1 \times 1$ Cuboids along with 1 other Cuboid}
\label{sec:stacked-net}

\begin{figure}
\centering
\includegraphics[width = 0.8\linewidth]{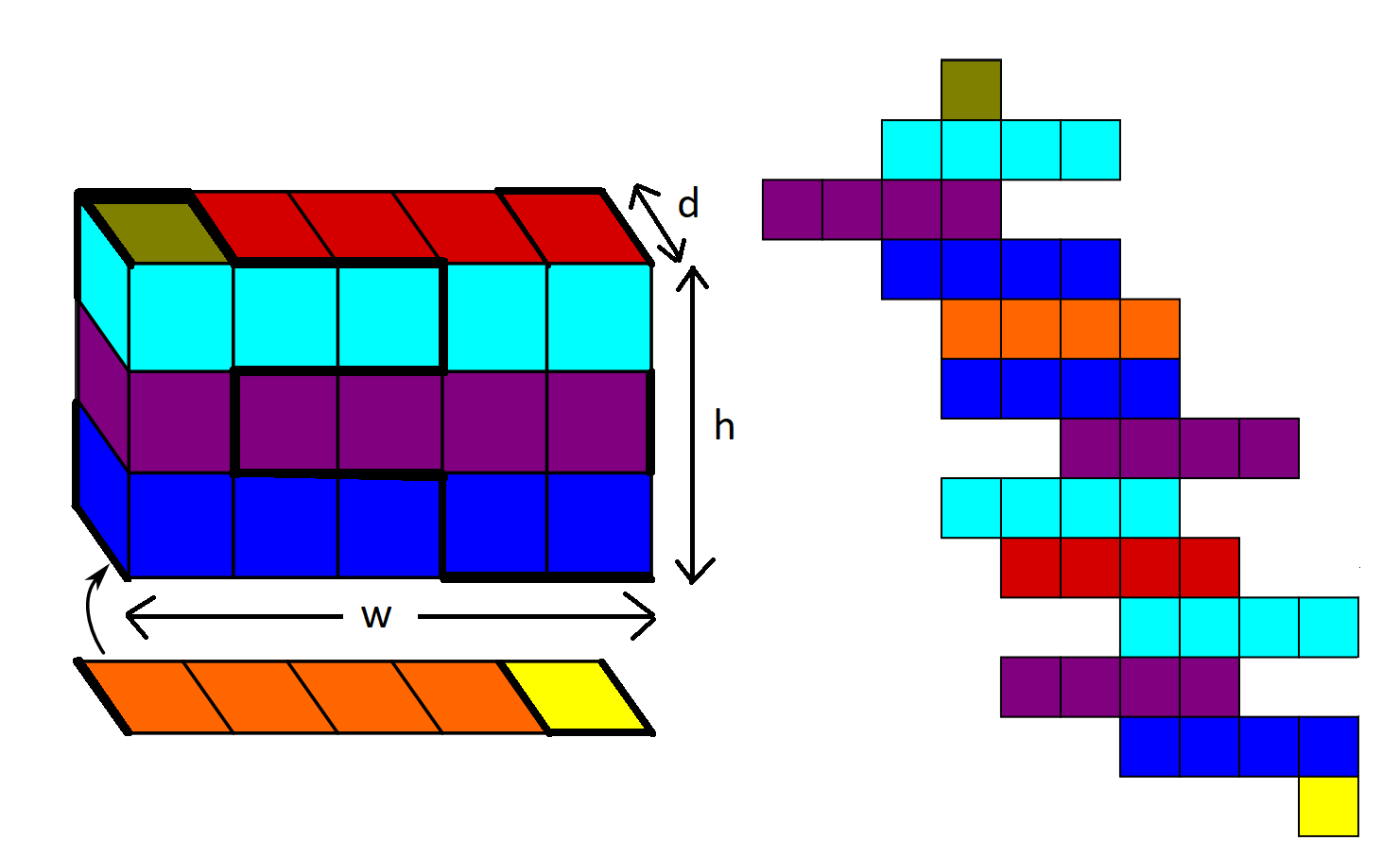}
\caption{Basic picture of a striped cuboid (left) along with the stacked $N\times 1 \times 1$ net (right). The thick black lines on the cuboid represent the cut lines for unfolding into the net.}
\label{fig:stripe_diagram}
\end{figure}

In this section, we will start with results for common unfoldings of 2 orthogonal boxes. For the numbers we were able to enumerate in Appendix~\ref{stacked_results_table}, we can verify that the combinations with a large number of nets are generally of the forms $$N\times 1 \times 1\ \&\ M\times (4k+3)\times 3$$ $$\mbox{or}$$ $$N\times 1 \times 1\ \&\ M \times (4k+1) \times 1.$$ For cuboids with these dimensions, over 95\% (but never 100\%) of the nets enumerated in Appendix~\ref{stacked_results_table} are \defn{striped nets}: a stacked net where if we paint the stacks in a certain pattern, the non-stacked cuboid it folds into gets ringed with that striped pattern. (\defn{striped cuboids} are the cuboids that have the striped pattern and \defn{striped rings} are the rings around the cuboid with a specific colour). Figure~\ref{fig:stripe_diagram} provides an illustration of the colouring. Striped cuboids have a variety of properties that we have proven in Appendix~\ref{striped_net_proofs}. Here, we highlight the following theorem:

\begin{theorem}(\textbf{The equation is $f(h)\approx(p/q)7^h$})\\
Assuming that the number of non-striped rings on the sides of the striped cuboid remains constant,
the number of solutions for a striped cuboid using the $N \times 1 \times 1$ cuboid as the layering cuboid is $\approx (p/q)7^h$ where $h$ is the height $h$ and $p/q$ is a rational number.
\end{theorem}
\begin{proof}
See Appendix~\ref{striped_net_proofs}
\end{proof}

For cases where the cuboids are not of forms $N\times 1 \times 1\ \&\ M\times (4k+3)\times 3$ or $N\times 1 \times 1\ \&\ M \times (4k+1) \times 1$, the number of nets is not as well understood and seems to consistently be above 0. We have enumerated our results in Appendix~\ref{stacked_results_table}.

\section{Stacked Nets with Area 106}

Previously, the smallest known net that folds into 3 distinct cuboids with a positive volume had an area of 532~\cite{magic_lid}. Here, we present new nets with area 106 that fold into 3 distinct cuboids with positive volume.

\begin{theorem}
    There exist 15 common stacked nets for the 1~$\times$~1~$\times$~26, 1~$\times$~2~$\times$~17, and 1~$\times$~5~$\times$~8 cuboids with surface areas of 106.
\end{theorem}

\begin{figure}
\includegraphics[width = \columnwidth, scale=0.4]{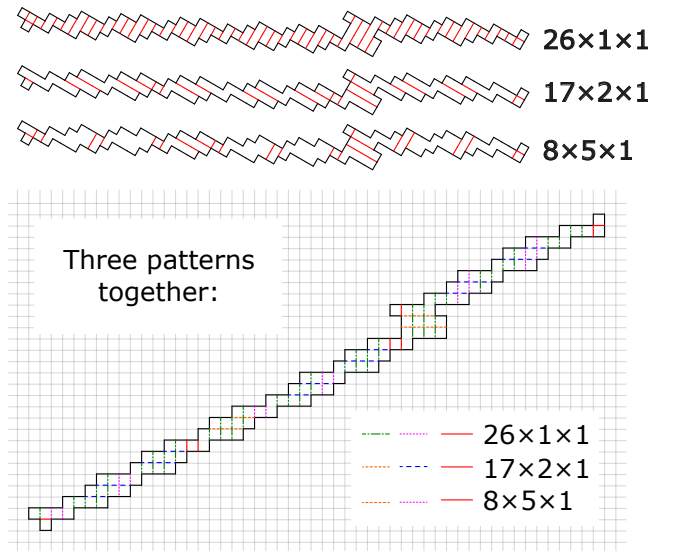}
\caption{Sample net of a common unfolding for the cuboids with a surface area of 106}
\label{fig:area 106}
\end{figure}

\begin{proof}
    These are enumerated in \texttt{Figs/stacked\_net\_solution}, where there is a text file describing each solution and pictures of a single solution covering three cuboids. An additional observation is that 2 of them are rotationally symmetric. One of these non-rotationally symmetric solutions is presented in Figure~\ref{fig:area 106}.
\end{proof}

\section{Generalization of the Area 106 nets}

We can now apply stripes introduced in Section~\ref{sec:stacked-net} to the nets of area 106 which can fold onto the surface of three different cuboids with positive volume. These stripes allow us to generalize our results with conjectures~\ref{conj:depth1} and~\ref{conj:depth3}.

\textbf{Definitions:}
\begin{itemize}
\item To condense the writing, let $p(C, D)$ be defined as the perimeter of a rectangle with length C and width D.
\item The stacks of a striped cuboid go around the depth $d$ and height $h$, and the stripes that go around the cuboid go around the width $w$ and the depth $d$. These dimensions are labeled in Figure~\ref{fig:stripe_diagram}.
\end{itemize}

\begin{conj}
\label{conj:depth1}
    Three cuboids with dimensions $N\times 1 \times 1$, $(4q+1)\times 1\times h_1$, and $(4r+1)\times 1\times h_2$ (where $q > r$) have a common striped net if the following conditions are met simultaneously:
    \begin{itemize}
        \item $h_1$ reached a minimum height dependent on $q$ and $r$.
        \item $p(4q+1, 1)/p(4r+1, 1)$ is an odd number
        \item $p(4q+1, 1)p(h_1, 1)=p(4r+1, 1)p(h_2, 1)$
        \item The surface area of $N\times 1 \times 1$ is made to match the surface area of the other cuboids.
    \end{itemize}

\end{conj}

See Appendix~\ref{striped_net_proofs} for theorems and proofs that support this conjecture.

We also found a similar thing for depth 3 cuboids. Specifically, we have the following conjecture:

\begin{conj}
\label{conj:depth3}
    Three cuboids with dimensions $N\times 1 \times 1$, $(4q+3)\times 3\times h_1$, and $(4r+3)\times 3\times h_2$ (where $q > r$) have a common striped net if the following conditions are met simultaneously:
    \begin{itemize}
        \item $h_1 > 2$, $h_2 > 2$, and $h_1$ reached a minimum height dependent on $q$ and $r$.
        \item $p(4q+3, 3)/p(4r+3, 3)$ is an odd number
        \item $p(4q+3, 3)p(h_1, 3)=p(4r+3, 3)p(h_2, 3)$
        \item The surface area of $N\times 1 \times 1$ is made to match the surface area of the other cuboids.
    \end{itemize}
\end{conj}

Although not fully proven, these conjectures help reveal locations in the $w$, $h$, and $d$ space where solutions are more densely populated. Following the conjectures, Tables~\ref{tab:depth-1-results} and~\ref{tab:depth-3-results} list the number of solutions that we found for common nets between an $N\times 1\times 1$ cuboid and cuboids of depth 1 and 3, respectively.

\begin{threeparttable}
\begin{tabular}{||c | c| c||} 
 \hline
 Cuboid Shape & \makecell{Number of three-way striped solutions\\ with $M\times 5\times 1$ and $N\times 1\times 1$} \\ \hline
 17 $\times$ 2 $\times$ 1 & 15 \\ \hline
 17 $\times$ 3 $\times$ 1 & 161 \\ \hline
 17 $\times$ 4 $\times$ 1 & 1,387 \\ \hline
 17 $\times$ 5 $\times$ 1 & 10,884 \\ \hline
 17 $\times$ 6 $\times$ 1 & 82,794  \\ \hline
 17 $\times$ 7 $\times$ 1 & 613,862 \\ \hline
\end{tabular}
\begin{tabular}{||c | c|} 
 \hline
 Cuboid Shape & \makecell{Number of three-way striped solutions\\ with $M\times 9\times 1$ and $N\times 1\times 1$}  \\ \hline
 29 $\times$ 2 $\times$ 1 & 0 \\ \hline
 29 $\times$ 3 $\times$ 1 & 11  \\ \hline
 29 $\times$ 4 $\times$ 1 & 246  \\ \hline
 29 $\times$ 5 $\times$ 1 & 2,854  \\ \hline
\end{tabular}
 \caption{Results for the number of three-way striped solutions for $N\times 1\times 1$ and cuboids with depth 1.}
 \label{tab:depth-1-results}
 \end{threeparttable}

\begin{threeparttable}
\begin{tabular}{||c | c| c||} 
 \hline
 Cuboid Shape & \makecell{Number of three-way striped solutions\\ with $M\times 3\times 3$ and $N\times 1\times 1$} \\ \hline
 3 $\times$ 15 $\times$ 3 & 12 \\ \hline
 4 $\times$ 15 $\times$ 3 & 231   \\ \hline
 5 $\times$ 15 $\times$ 3 & 2,092 \\ \hline
 4 $\times$ 47 $\times$ 3 & 231 \\ \hline
 5 $\times$ 47 $\times$ 3 & 2,092 \\ \hline
\end{tabular}
\end{threeparttable}
\\
\\
\begin{threeparttable}
\begin{tabular}{||c | c| c||} 
 \hline
 Cuboid Shape & \makecell{Number of three-way striped solutions\\ with $M\times 7\times 3$ and $N\times 1\times 1$} \\ \hline
 4 $\times$ 27 $\times$ 3 & 8 \\ \hline
 5 $\times$ 27 $\times$ 3 & 264   \\ \hline
 4 $\times$ 47 $\times$ 3 & 8 \\ \hline
 5 $\times$ 47 $\times$ 3 & 392 \\ \hline
\end{tabular}
  \caption{Results for the number of three-way striped solutions for $N\times 1\times 1$ and cuboids with depth 3.}
 \label{tab:depth-3-results}
\end{threeparttable}

The files in \texttt{Figs/semi\_striped\_net\_solution} show a net covering the three cuboids with dimensions $58\times 1\times 1$, $3\times 3\times 18$, and $15\times 3\times 4$. We chose those dimensions because they are the smallest dimensions that work for depth 3 while satisfying the requirement that all of the cuboids are of different dimensions.

Note that because the number of stacked nets covering the three cuboids of dimension: $26 \times 1 \times 1$,  $8\times 5\times 1$, and $17\times 2\times 1$, matches the number of striped nets covering those dimensions, all stacked nets covering 3 cuboids of area 106 are striped nets.

\section{Simple-Phase nets and the Matrix Equation}

After finding stacked nets that cover three cuboids, we decided to slightly widen the search from stacked nets to $N\times 1 \times 1$ nets where the top and bottom $1\times 1$ tiles of the $N\times 1 \times 1$ cuboid only have one neighbor (\emph{i.e.} the \defn{simple-phase nets}).

After some experimentation, it was observed that every 4 tile 'band' or 'layer' in between the top and bottom $1\times 1$ sides could only be in one of seven 'layer states', and whether or not state $i$ could exist on top of state $j$ depends only on state $j$. To see a picture that is a simple-phase net that contains all seven 'layer states', see Figure \ref{simple-phase net all 7}.

\begin{figure}
\centering
\includegraphics[width=50mm,scale=1.0]{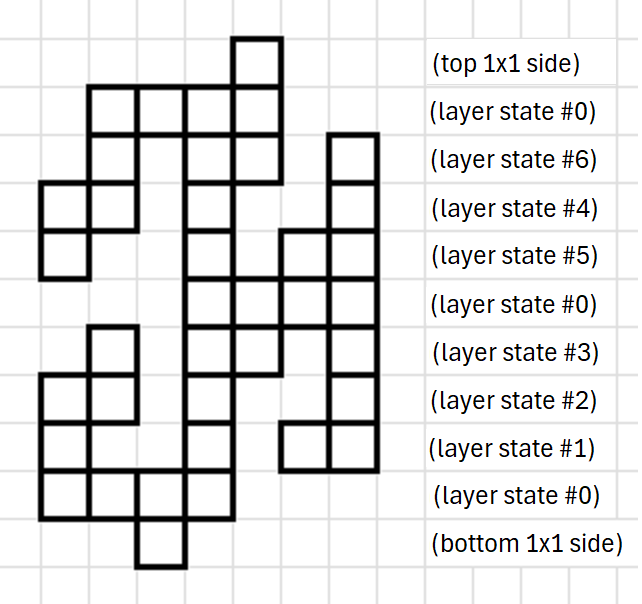}
\caption{Example of a simple-phase net that uses all 7 layer states. They are distinct from each other because of the different order in which tiles appear starting from the leftmost tile, and by how the different 'tile islands' connect with each other from underneath the layer state.}
\label{simple-phase net all 7}
\end{figure}

These observations led to the idea that the enumeration of this type of solution, while counting symmetric solutions as distinct, could be counted via a matrix equation. The way to understand the matrix is by reading the tile in row $i$ and column $j$ as "the number of ways we could add layer state $i$ onto layer state $j$". Also note that the first layer must be layer state 0 because all the other layer states will need at least 1 layer below it to fully connect them, and by symmetry, the same argument works for the last layer.

Here is the matrix equation with all 7 states, where $n$ is the number of layers and $n\geq1$:
$$
F(n) = 16
\begin{pmatrix}
1\\
0\\
0\\
0\\
0\\
0\\
0\\
\end{pmatrix}^\intercal
\begin{pmatrix}
7 & 2 & 1 & 2 & 2 & 1 & 2\\
1 & 1 & 1 & 1 & 0 & 0 & 0\\
2 & 1 & 1 & 1 & 0 & 0 & 0\\
2 & 1 & 1 & 1 & 0 & 0 & 0\\
1 & 0 & 0 & 0 & 1 & 1 & 1\\
2 & 0 & 0 & 0 & 1 & 1 & 1\\
2 & 0 & 0 & 0 & 1 & 1 & 1\\
\end{pmatrix}
^{n-1}
\begin{pmatrix}
1\\
0\\
0\\
0\\
0\\
0\\
0\\
\end{pmatrix}
$$

The largest eigenvalue is approximately equal to 9.4956 and is the real solution to the following equation: $\lambda^3-10\lambda^2+5\lambda-2=0$. That means that every newly added layer will increase the number of nets by close to a multiple of that amount. Armed with this information, we changed the stack search algorithm to accommodate these 7 types of layers.

\subsection{Basic Observations of the Results}
\begin{enumerate}
\item The number of simple-phase nets appear to be a scaled-up version of the number of stacked nets.
\item We were unable to verify whether the area-102 cuboids ($25\times 1\times 1$) have a simple-phase net that covers three cuboids. Fortunately, even though we could not find the number of results for '$23\times 1\times 1$ \& $7\times 5\times 1$' and '$26\times 1\times 1$ \& $8\times 5\times 1$', we were still able to check if there was a net that covers at least 3 cuboids by checking the results of the other runs against the $7\times 5\times 1$ cuboid and the $8\times 5\times 1$ cuboid.
\item See Appendix \ref{stacked_results_table} for the full table of results.
\end{enumerate}
\subsection{The 3-way Solutions}
We found a total of 40 simple-phase nets of area 106 that cover the $26\times 1\times 1$ \& $17\times 2 \times 1$ \& $8\times 5\times 1$ cuboids, and 4 of those 40 are rotationally symmetric. In other words, we found 25 new nets that are not stacked nets. See the files in Figs/simple\_phase\_solution for the ASCII representation of the solutions. An interesting fact about these solutions is that after visually inspecting them, we conjecture that all of the $17\times 2 \times 1$ \& $8\times 5\times 1$ cuboid solutions are striped if you give the layers of the $26\times 1\times 1$ cuboid a certain colour pattern.

\section{Thorough Search for Nets}

In this section, we have the results of using a thorough algorithm to find nets. Because that algorithm turns out to be 3 to 10 times slower than what was done by Qian et al. in 2025~\cite{qian2025unfoldingboxeslocalconstraints}, it did not update the lowest possible area a net that covers three cuboids could have. Despite that, we believe this section still has value because the algorithm used is very different from their SAT solver algorithm, and because we found non-standard results.

\section{How the Net Searches Were Done}

\subsection{Exploration with the Breadth-First Search Iterator Algorithm}

The idea of this algorithm is to search for nets by trying every possible way in which a breadth-first search algorithm can explore a net in the shape of a polyomino. If we do that, we will have searched all non-overlapping nets. Here is a simple version of the algorithm that does not include wrapping around a model of the cuboid:

\definecolor{codegreen}{rgb}{0,0.6,0}
\definecolor{codegray}{rgb}{0.5,0.5,0.5}
\definecolor{codepurple}{rgb}{0.58,0,0.82}
\definecolor{backcolour}{rgb}{0.95,0.95,0.92}

\begin{lstlisting}[language=Python, basicstyle=\tiny, label=BFS Iterator,title=BFS Iterator,backgroundcolor=\color{backcolour},     commentstyle=\color{codegreen},
    keywordstyle=\color{magenta},
    numberstyle=\tiny\color{codegray},
    stringstyle=\color{codepurple},]
net_search_start(max_size):
    Define Mapping <Int to Int> tile_to_ordering
    cur_size = 0, cur_order_index = 0, cur_rot_index = 0
    tile_to_ordering.put(cur_order_index, cur_size)
    insert_initial_tile()
    cur_size = cur_size + 1

    net_search(tile_to_ordering, max_size, cur_size,
    cur_order_index, cur_rot_index)
    
net_search(Mapping <Int to Int> tile_to_ordering,
    max_size, cur_size, cur_order_index, cur_rot_index):

    if max_size == cur_size:
        process_solution()
        return

    for i = cur_order_index to cur_size:

        curTile = ordering_to_tile(i)
        
        foreach rot from 0 to 3 inclusive:
            if i == cur_order_index and cur_rot_index > rot:
                # Skip iteration if rot before BFS queue
                continue
            
            neiTile = getNeiOnFlatPaper(curTile, rot)

            # If neiTile has a neighbour that was inserted
            # before curTile, it was already explored.
            if neighbouring_location_non_empty() or
               cur_tile_being_considered_was_already_explored
                (curTile, neiTile, tile_to_ordering, i, rot):
                continue

            # Try inserting neighbour:
            insert_tile(curTile, neiTile, rot)
            tile_to_ordering.put(neiTile, cur_size)
            cur_size = cur_size + 1

            # Recursive call:
            Net_search(tile_to_ordering, max_size,
            cur_size, i, rot)
            
            # Undo inserting neighbour:
            cur_size = cur_size - 1
            remove_tile(neiTile)
            tile_to_ordering.remove(neiTile)
\end{lstlisting}

To summarize the code above, at every level, this algorithm tries every way the BFS algorithm could add the next tile by adding it, pretending that it was the next added in the BFS algorithm and going one level deeper into the recursion. It has the benefit of not needing to save previous configurations. This means that we can easily check trillions of configurations without worrying about memory overhead. For an example of a specific polyomino solution from the previous algorithm if the size was set to 10, see Figure \ref{fig:bfs_iterator}.

\begin{figure}\label{BFS-iterator-drawing}
\centering
\includegraphics[width=50mm,scale=0.5]{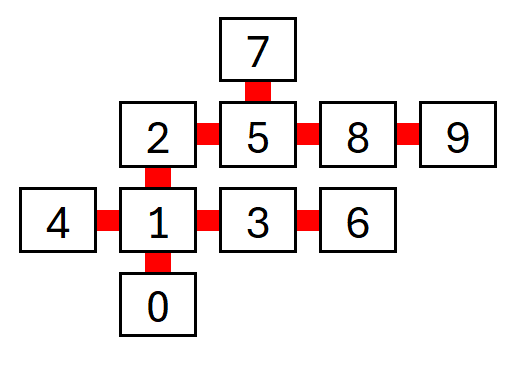}
\caption{If we have the breadth-first search algorithm start at the bottom, the numbers on the tiles indicate the order in which the tiles are found, and the red links indicate from where each tile was found. It is also the order and links in which the tiles are recursively found in the BFS-iterator algorithm for this polyomino with 10 tiles. Any other ordering for this polynomino starting from tile index 0 will be blocked at some point.}
\label{fig:bfs_iterator}
\end{figure}

\subsection{The Connection to Redelmeier's Algorithm}

The BFS iterator algorithm and Redelmeier's algorithm~\cite{Redelmeier-Counting-Polyominoes} are closely linked. A shallow link to Redelmeier's algorithm is that both algorithms could enumerate the number of polyominoes with $n$ tiles without exponential memory overhead, but there is a deeper link.
The deeper link is the observation that, if one wished, one could direct Redelmeier's algorithm to count the polyominoes in the same order as the BFS iterator algorithm. If you interpret Redelmeier's algorithm's untried set as the currently unexplored neighbours of the BFS iterator algorithm, and remove the next element of the untried set as you would add a new neighbouring tile with the BFS iterator algorithm, we have noticed that the order in which the algorithms find polyominoes will match.

\subsection{Other Applications to the Breadth-First-Search Iterator}

We feel that the BFS iterator algorithm is the hidden gem of this paper. The simplicity and effectiveness of the algorithm suggest that it should be the baseline solution to many counting problems in which there are some indistinguishable pieces that are stuck together. To be more specific, example problems this algorithm could tackle include enumerating polycubes, enumerating 3D nets covering hypercuboids, finding solutions to the "Cubigami 7" puzzle originally developed by Knuth and Miller, enumerating the ways to make a bigger shape given a set of smaller pieces, etc...

\subsection{Including Cuboid Nets}
To get this algorithm to work with cuboids, we just had to build the cuboid nets as the algorithm progresses, while making sure that unrelated edges of the cuboid do not touch each other. This algorithm requires all cuboids to have a starting tile and rotation for the root tile. To be thorough, we could fix the start location and start rotation for the first cuboid, but then we cannot make assumptions about the next cuboids, so we will have to try every symmetrically distinct combination of start locations and rotations for all of the cuboids except for the first one.

\subsection{Results}

\subsubsection{$N\times 1 \times 1$ Solutions}

\begin{table}[]
    \centering
\begin{tabular}{|c | c |c |} 
\hline
N & Num Nets & New?\\
\hline
1 & 11 & No\\
\hline
2 & 723 & No \\
\hline
3 & 14,978 & No \\
\hline
4 & 228,547 & No \\
\hline
5 & 3,014,430 & Yes \\
\hline
6 & 36,175,381 & Yes \\
\hline
7 & 407,023,305 & Yes \\
\hline
\end{tabular}
    \caption{List of the number of nets (no cuts) that fold to a $N\times 1 \times 1$ cuboid}
    \label{tab:Nx1x1-list}
\end{table}

Table~\ref{tab:Nx1x1-list} lists the number of normal nets (\textbf{without invisible cut between edges}) that fold the $N\times 1 \times 1$ cuboid. For $N\le4$, we match the results found by Qian et al.~\cite{qian2025unfoldingboxeslocalconstraints}. The rest of the results are new results.

For the $N=6$ and $N=7$ cases, we had to avoid a memory overflow by having the algorithm look at each solution and logically deduce whether it is a duplicate version of what was seen before instead of just saving all of the solutions to a hash table and doing a lookup.

\subsubsection{$N\times 1 \times 1$ and Other Cuboid(s) Solutions}
\label{sec:small_area_results}

We reproduced the results in Qian et al.~\cite{qian2025unfoldingboxeslocalconstraints} for cuboids of area 46 or less. 
The only new result we have is that  we found that the number of nets (\textbf{without invisible cut between edges}) covering the $11\times 1\times 1$ cuboid and the $5\times 3\times 1$ cuboid is 4,233,426. It took around 3 weeks of computation to get this result.

\subsection{Non-standard Solutions}

Historically, there are two main ways to define a polyomino net, and they lead to different counts. The first way insists that if the edges of the net are touching, they stick together. The second way to count is to allow edges of the net to touch while not actually being stuck together if having them stuck together stops the net from covering the cuboids. According to Mitani and Uehara's paper \cite{mitani2009polygons}, necessary cuts in the net could happen if there is an apparent 'hole' in the net.

For counting nets using the second way to count, if you want to match the 1080 result in Dawei Xu et al.~\cite{Xu-Uehara-area-30} and the edge touching results in Qian et al.~\cite{qian2025unfoldingboxeslocalconstraints}, you need to follow these two unstated rules: the necessary cuts in the net must be the same for all cuboids involved, and the information about where the cuts are located will not be used to distinguish two nets. Just for fun, we decided to break these rules so that we could see what would happen. When we forced a net with a hole to need one way to cut the net to cover the $5\times 1\times 1$ cuboid, and another way to cut the net to cover the $3\times 2\times 1$ cuboid, we found 36 new solutions. See  the Figs/nets\_needs\_2\_ways\_to\_cut/ folder for the full set of solutions. When we looked for nets with holes where there are at least two valid ways to cut the net before being able to cover both the $3\times 2\times 1$ cuboid and the $5\times 1\times 1$, we found 4 nets with holes with 2 distinct ways to cut it up, none for the $3\times 3\times 1$ and $7\times 1\times 1$ cuboids, and 3 for the $5\times 2\times 1$ and $8\times 1\times 1$ cuboids. See the Fig/nets\_having\_two\_ways\_to\_cut folder for these results. The existence of such nets answers the open question posed in Section 3.1.1 of Ryuhei Uehara's book~\cite{intro_comp_origami}.

\section{Open Questions Raised While Studying The Results}

\subsection{Can the simple-phase search algorithm for $N\times 1 \times 1$ cuboids algorithm be extended to include all $N\times 1 \times 1$ nets?}

While doing the simple-phase search, we found that the search can be seen as a depth-first search where the types of layer you could add on are mapped out by a $7\times 7$ matrix. We conjecture that the full search of all polyomino nets covering the $N\times 1 \times 1$ cuboids could also be mapped out by a bigger matrix. Once the matrix is found, we could potentially use it to speed up the search for nets that fold into the $N\times 1 \times 1$ cuboid and other cuboids. We also conjecture that the top eigenvalue of that matrix will match the top eigenvalue of the $7\times 7$ matrix, which is approximately 9.496... Realistically, the matrix might be an order of magnitude bigger than the $7\times 7$ matrix, so deriving that matrix may not be easy. We feel that once it is achieved, we can go further and have a matrix for other $N\times A\times B$ cuboids where A and B are constants and N is a variable.

\subsection{Are there two cuboid shape with the same surface area, but no common nets?}\label{sec:no-common-net}

Every time we checked to see if there were nets that covered the $N\times 1 \times 1$ cuboid and some other cuboid, we found nets. We are not sure whether this remains the case for any two cuboids sharing the same area. For this reason, we think it would be interesting if the $5\times 5 \times 5$ cuboid does not share a net with the $37\times 1\times 1$ cuboid. More generally, we are open to the possibility that if the minimum face-perimeter of one shape is lower than all three dimensions of the other shape, then we will not be able to find common nets. Qian et al. also investigated this question and found that the area of the smallest area cuboids not sharing a net must be greater than 86 \cite{qian2025unfoldingboxeslocalconstraints}. We hope that we may find a human-readable proof of the non-existence of common nets. Although it may be dubious, we are encouraged by the fact that we do not need a computer to prove that the $0\times 1 \times n$ 'cuboid' where $n>10$ cannot fold into a doubly-covered T-shape. 

\section{Conclusion}

The search for nets that cover 3 or 4 cuboid shapes, and the results we found along the way, has opened up many different and interesting questions that will hopefully be explored in future research. We especially look forward to getting answers to the two main questions that remain open:

\begin{enumerate}
\item  What is the smallest possible net by unit area in which a polyomino could be used to cover 3 cuboid shapes with orthogonal folds?
\item  Is there a polyomino net that covers 4 cuboid shapes with orthogonal folds?
\end{enumerate}

Getting an answer to the first question seems somewhat promising because we only need to search for nets of area 62 (\emph{i.e.} $7\times 3\times 1$, $5\times 3\times 2$, and $15\times 1\times 1$) up to area 102 (\emph{i.e.} $7\times 3\times 3$, $9\times 3\times 2$, $12\times 3\times 1$, and $25\times 1\times 1$), and we may just need more compute, faster algorithms, and a generalization of the simple-phase algorithm. Although it may not be enough, we feel that we have a plan to follow. Getting an answer to the second question is more rewarding, but also seems to be more challenging because it will probably require a creative leap which comes from a novel idea or a creative recombination of what is already known.


\small
\bibliographystyle{abbrv}

\newpage
\section*{Appendix}

\renewcommand{\thesubsection}{\Alph{subsection}}

\subsection{Enumeration of Stacked and Simple-Phase $N\times 1 \times 1$  nets also Covering 1 Other Cuboid}\label{stacked_results_table}

\begin{tabular}{||c | c | c | c | c||} 
 \hline
 \makecell{Main\\Cuboid}&\makecell{Other\\Cuboid} &\makecell{Number of\\Stacked Nets}& \makecell{Number of\\Simple-Phase\\Nets}\\ [0.5ex] 
 \hline\hline
 5x1x1 & 3x2x1 & 7 & 34 \\ \hline
 7x1x1 & 3x3x1 &  2 & 6 \\ \hline
 8x1x1 & 5x2x1 & 109 & 197 \\ \hline
 9x1x1 & 4x3x1 & 12 & 33 \\ \hline
 10x1x1 & 3x3x2 & 133 & 158 \\ \hline
 11x1x1 & 5x3x1 & 591 & 1,401 \\ \hline
 11x1x1 & 7x2x1 & 19 & 54 \\ \hline
 13x1x1 & 3x3x3 & 175 & 203 \\ \hline
 13x1x1 & 6x3x1 & 6 & 77 \\ \hline
 14x1x1 & 5x4x1 & 4,182 & 13,264 \\ \hline
 14x1x1 & 9x2x1 & 152 & 355 \\ \hline
 15x1x1 & 5x3x2 & 184 & 507 \\ \hline
 15x1x1 & 7x3x1 & 9 & 31 \\ \hline
 16x1x1 & 4x3x3 & 1,285 & 1,365 \\ \hline
 17x1x1 & 5x5x1 & 28,817 & 124,822 \\ \hline
 17x1x1 & 8x3x1 & 17 & 54 \\ \hline
 17x1x1 & 11x2x1 & 29 & 89 \\ \hline
 19x1x1 & 5x3x3 & 8,418 & 11,290 \\ \hline
 19x1x1 & 7x4x1 & 27 & 52 \\ \hline
 19x1x1 & 9x3x1 & 951 & 2,657 \\ \hline
 20x1x1 & 6x5x1 & 202,106 & 1,187,120 \\ \hline
 20x1x1 & 7x3x2 & 296 & 581 \\ \hline
 20x1x1 & 13x2x1 & 211 & 507 \\ \hline
 21x1x1 & 10x3x1 & 11 & 122 \\ \hline
 22x1x1 & 5x5x2 & 24 & 163 \\ \hline
 22x1x1 & 6x3x3 & 58,891 & 105,801 \\ \hline
 23x1x1 & 5x4x3 & 61 & 222 \\ \hline
 23x1x1 & 7x5x1 & 1,411,798 & ?? \\ \hline
 23x1x1 & 11x3x1 & 15 & 51 \\ \hline
 23x1x1 & 15x2x1 & 28 & 150 \\ \hline
 24x1x1 & 9x4x1 & 6,853 & ?? \\ \hline
 25x1x1 & 7x3x3 & 410,329 & ?? \\ \hline
 25x1x1 & 9x3x2 & 64 & ?? \\ \hline
 25x1x1 & 12x3x1 & 28 & 81 \\ \hline
 26x1x1 & 17x2x1 & 268 & 674 \\ \hline
 26x1x1 & 8x5x1 & 9,885,286 & ?? \\ \hline
 27x1x1 & 13x3x1 & 1,310 & ?? \\ \hline
 27x1x1 & 5x5x3 & 154 & ?? \\ \hline
 27x1x1 & 7x6x1 & 50 & ?? \\ \hline
 28x1x1 & 8x3x3 & 2,870,327 & ?? \\ \hline
 29x1x1 & 7x5x2 & 165 & ?? \\ \hline
 29x1x1 & 11x4x1 & 36 & ?? \\ \hline
\end{tabular}

\subsubsection{Commentary on the Enumeration Table}

We only understand the counts for the combinations of dimensions described in Section \ref{sec:stacked-net}. The reason we shared the full list is to encourage and help others develop stronger algorithms that could verify these numbers and then possibly verify Section \ref{sec:no-common-net}'s claim that there is no net that covers both the $5 \times 5 \times 5$ and $37 \times 1 \times 1$ cuboids.

\subsection{Proofs and Information Related to Striped Cuboids}\label{striped_net_proofs}

\subsubsection{How to create the 3x(3+4k)xN Striped Cuboid}
The $M\times 3\times 3$ cuboids could also have a striped pattern, but the top and bottom $3\times 3$ sides are slightly more complicated. In the $M\times 3\times 3$ case, the bottom $1\times 1$ tile starts at the first layer above the bottom $3\times 3$ side, stacks up through the top $3\times 3$ side only to go back down through bottom $3\times 3$ side, and then it goes up to the $1\times 1$ tile on the layer right below the top $3\times 3$ side. Because it is periodic, the layers between the first layer and last layer above the $3\times 3$ side could have a striped pattern. This observation could also be generalized to all cuboids of the form $M\times (4k + 3) \times 3$. We just need to know that as $k$ grows, the number of times the $N\times 1 \times 1$ stacks go around the $M\times 3$ side grows with it.

\subsubsection{Definitions}
\begin{itemize}
    \item The layering cuboid is the long cuboid in which one of its face perimeters defines the width of the stacks in the stacked net. In this paper, the layering cuboid is always the $N\times 1\times 1$ cuboid.
    \item An $L\times 1$ layer represents one of the stacks of the layering cuboid. If the layering cuboid is the $N \times 1 \times 1$ cuboid, then the width is the perimeter of the $1 \times 1$ side, which is 4.
    \item The shift amount is the amount of tiles that a $L\times 1$ layer shifts sideways compared to a previous $L\times 1$ layer.
    \item The global shift amount is the total amount of tiles shifted sideways from the first striped ring to the last striped ring reduced modulo the perimeter of $w \times d$.
    \item The initial shift amount is the amount of tiles horizontally to right of the top left corner of the $h \times w$ side the first $L \times 1$ layer to be completely inside of the $h \times w$ is. By this definition, the possible values only range from $0$ to $(L-1)$.
\end{itemize}

\begin{lemma}
    The equivalent of the initial shift amount, but for the bottom left is determined by the global shift amount and the initial shift amount.
\end{lemma}
\begin{proof}
By the definitions, the shift amount of the bottommost striped ring could be determined by adding the global shift amount to the initial shift amount modulo the perimeter of $w\times d$ and then reducing that modulo the layer width $L$.
\end{proof}

\begin{theorem}\label{shift_theorem} \textbf{(Unique Shift Between Striped Rings Theorem)}\\
There is only one unique 'shift amount' for all the $L\times 1$ layers connected to each other while being part of 2 striped ring.
\end{theorem}
\begin{proof}
If we focus on how two $L\times 1$ layers connect between two striped rings, we see that it will force the two $L\times 1$ layers adjacent to it to connect the same way, and that will force the next two, and so on all the way around.\end{proof}

\begin{lemma}\label{shift_growth_theorem}
There is $2L - 1$ ways two striped rings could be connected to each other.
\end{lemma}
\begin{proof}
If we focus on how two $L\times 1$ layers connect between two striped rings, then we could count that there are $2L-1$ ways a $L\times 1$ layer could be placed above the one below it, and because of Theorem \ref{shift_theorem}, the same shift amount will be repeated for the other $L\times 1$ layers on those striped rings.

\end{proof}

\begin{lemma}\label{global_shift}
The top and bottom parts of a striped cuboid only depend on the amount the striped rings shift reduced modulo the perimeter of $w \times d$ and how much the first striped ring is shifted compared the top part.
\end{lemma}
\begin{proof}
Observe that as far as the top and bottom parts of a striped cuboid are concerned, if the striped rings part of the striped cuboid does a full rotation around clockwise or anti-clockwise, the effective twist will be reduced by modulo perimeter of $w \times d$.
The same logic applies if the middle part rotates around $k$ times.
\end{proof}

\begin{lemma}\label{even_dist_global}\textbf{(Evenly distributed global shift)}\\
Given a set of striped rings, the total amount of twisting due to the striped rings is evenly distributed between all possibilities modulo the perimeter of $w \times d$ as $h$ increases.
\end{lemma}

\begin{proof}
If we imagine that every shift amount between the striped rings is random, the variance of the sum of the shift amounts for each pair of striped rings becomes much greater than the perimeter of $w \times d$ as $h$ increases. This means that as $h$ increases, the
total amount of twisting from the top to the bottom modulo the perimeter of $w \times d$ is effectively random. Therefore, it must evenly distribute between all possibilities modulo the perimeter of $w \times d$ as $h$ increases.
\end{proof}

\begin{lemma}\label{symmetry} As the height of the striped cuboid increases, most striped nets that cover a striped cuboid will have 3 other symmetrically identical ways of covering the cuboid.
\end{lemma}
\begin{proof}
Given some height $h$, let $F(h)$ be the number of striped net solutions where the symmetries are ignored and let $f(h)$ be the number of symmetrically distinct striped net solutions.

Assuming a stacked net solution does not have any rotational or mirror symmetries, the stacked net has 3 other symmetrically identical combinations where the stacks remain horizontal, giving four possible fixed solutions for one free solution. Because the cuboid the stacked net folds into has rotational and mirror symmetries, those 4 symmetries could cover the cuboid without needing to be rotated or flipped, and because the vast majority of stacked nets are not symmetric as $h$ increases, the approximation $f(h)\approx F(h)/4$ is appropriate.
\end{proof}

\begin{theorem}(\textbf{The equation is $f(h)\approx(p/q)7^h$})\\
Assuming that the number of non-striped rings on the sides of the striped cuboid remains constant,
the number of solutions for a striped cuboid using the $N \times 1 \times 1$ cuboid as the layering cuboid is $\approx (p/q)7^h$ where $h$ is the height $h$ and $p/q$ is a rational number.
\end{theorem}
\begin{proof}
For convenience, let the perimeter of $w \times d = p(w,d)$.
 
Let us focus on building the striped solutions, while ignoring symmetries, in 2 parts where part 1 counts the number of ways the net is built when the shift amounts do not involve going from a striped ring to another striped ring,
and part 2 involves counting the number of ways we can have a global shift amount between the striped rings.

Let $f_1(i, j)$ be the number of ways to obtain a solution for part 1 given a global shift amount of $i$ and an initial shift amount of j, and let $d_1(i,j)$ be the number of duplicate ways to obtain a solution for part 1.
Let $f_2(i, j)$ be the number of ways to obtain a solution for part 2 given a global shift amount of $i$ and an initial shift amount of j.

Because the two parts are independent for each global shift amount and initial shift amount, the total \# of solutions $= \sum_{i=1}^{p(w,d)}\sum_{j=0}^{3}(f_1(i, j)-d_1(i,j))  f_2(i, j)$.

Because of the fact that there are only 4 initial shift amounts possible (by definition of initial shift amount) and by applying Lemma \ref{shift_growth_theorem} for every pair of adjacent striped rings, we get that $\sum_{i=1}^{p(w,d)}\sum_{j=0}^{3}f_2(i, j) = 4(7^{n-1})$ where n is the number of striped rings. By Lemma \ref{even_dist_global}, the number of solutions for part 2 distribute evenly, so $f_2(i, j) = 7^{(n-1)}/(p(w,d))$ for all relevant $i$ and $j$ where $n$ is the number of striped rings.\\

Therefore, the total number solutions $=  \sum_{i=1}^{p(w,d)}\sum_{j=0}^{3} (f_1(i, j)-d_1(i,j)) (7^{(n-1)}/p(w,d))$ \\$ = (7^{(n-1})/p(w,d))\sum_{i=1}^{p(w,d)}\sum_{j=0}^{3} (f_1(i, j)-d_1(i,j)) $.

By the way $f_1(i, j)$ and $d_1(i,j)$ is defined, and because we assumed the non-striped rings do not change as $h$ increases, the sum could only be an integer constant, so let us call it $K$.

Therefore, the total \# of solutions $= K(7^{(n-1)})/(p(w,d))$

Because we assumed that the number of non-striped rings stays constant, the number of striped rings must increase as the height $h$ increases. Therefore, we could let $h = n + c$ for some constant $c$.

Therefore, the total \# of solutions $= K(7^{(h-c-1)})/p(w,d)$.
 $= 7^h(K/(p(w,d)7^{c+1})).$\\
 $= (p/q_0)7^h$ where $p$ and $q_0$ are integers.

As a final step, because of Lemma \ref{symmetry}, the number of symmetrically distinct solutions is approximately a quarter of $(p/q_0)7^h$, which could be represented by the equation $f(h)= (p/q)7^h$ where $q=4q_0$.
\end{proof}

\begin{theorem}\label{area_equiv_1}\textbf{(Area Equivalence Theorem)}\\
If we fix the depth to a constant $d$, and set the surface area of $h_1\times d\times w_1$ to be equal to the surface area of $h_2\times d \times w_2$, then $(h_1+d)(w_1+d)=(h_2+d)(w_2+d)$.
\end{theorem}
\begin{proof}
It is known that the surface area of a cuboid $= 2(wd + hd + wh) = 2(w+d)(h+d)-2d^2$.\\
Therefore,
surface\_area\_1 $= 2(w_1+d)(h_1+d)-2d^2 =$ surface\_area\_2 $= 2(w_2+d)(h_2+d)-2d^2$.\\
Therefore, $2(w_1+d)(h_1+d)-2d^2 = 2(w_2+d)(h_2+d)-2d^2$.\\
Therefore, $(w_1+d)(h_1+d) = (w_2+d)(h_2+d)$.
\end{proof}

 In an attempt to trivialize the formula, Figure \ref{perimeter_factor} shows half the sides of a flattened cuboid.

\begin{cor}\label{perimeter_mult} \textbf{(Perimeter Multiplication Corollary)}
If we fix the depth, and the surface area of $h_1 \times d \times w_1$ is equal to the surface area of $h_2\times d\times w_2$, then
$p(h_1, d)p(w_1,d)=p(h_2, d)p(w_2,d)$ where p(C, D) is the perimeter of $C\times D$.
\end{cor}
\begin{proof}
Just multiply the equation in Theorem \ref{area_equiv_1} by four and recognize and factor the perimeters in both sides of the equation.
\end{proof}

\begin{figure}
\centering
\includegraphics[width=80mm,scale=1.0]{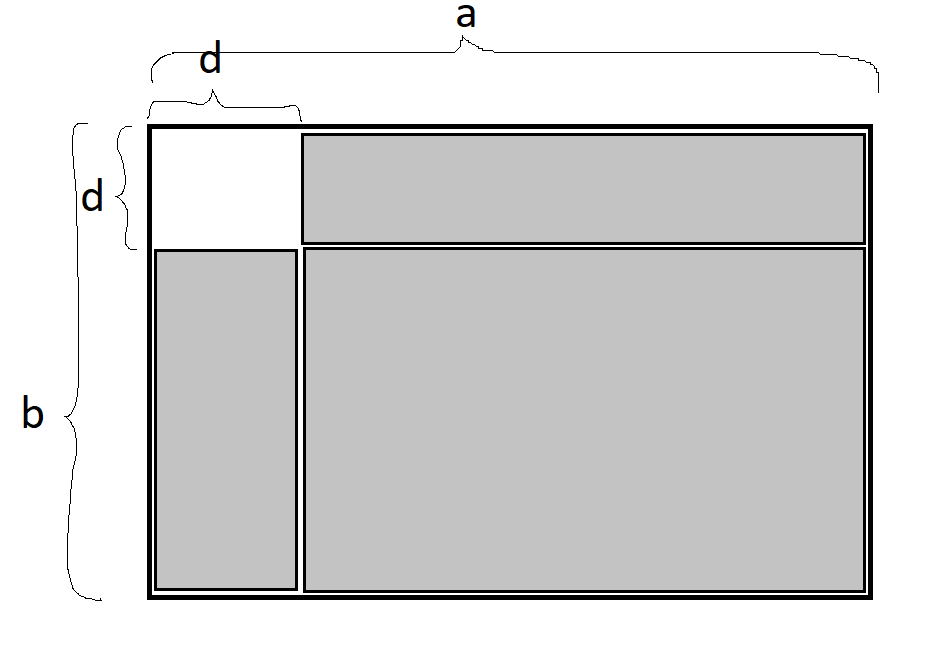}
\caption{Picture to accompany the perimeter factor argument. This figure is trying to show that the surface area of the cuboid is $2(ab-d^2)$ where $a=w+d$ and $b=h+d$.}
\label{perimeter_factor}
\end{figure}

\subsubsection{Perimetor Factor Analysis}
\textbf{Definition:}\\
A perimeter factor is the ratio between the perimeter of $d \times h_2$ of cuboid 2 and $d \times h_1$ of cuboid 1. This concept is useful because when the depth $d$ is the same, we can use Corrollary \ref{perimeter_mult} to deduce the dimensions of the second cuboid by knowing the dimensions of the first cuboid and the perimeter factor.

\begin{obs}\label{constant_cycle}\textbf{(Constant Number of Stacks to Go Around)}
The number of stacks required for a striped net to go around the striped cuboid 1 time is the perimeter of $h\times d$. This observation is based on the striped nets we have observed.
\end{obs}

\begin{theorem}\label{shift_constraint}\textbf{(The Shift Equivalence Constraint Theorem)}\\
Given a striped cuboid with two adjacent striped rings, due to Theorem \ref{shift_theorem}, and assuming observation \ref{constant_cycle}, the shift amounts between stacks associated with the shift amount for those two adjacent striped rings are predictable.
\end{theorem}
\begin{proof}

By Observation \ref{constant_cycle} and Theorem \ref{shift_theorem}, if we are dealing with two adjacent striped rings, the shift amount will repeat every $2(h+d)$ layers (or every full rotation around the $h\times d$ rectangle).

On top of that, the shift amount will also repeat somewhere in between because if the layers were originally going up the rings, they will also need to go down the rings. In other words, if we consider the shift-amount index $d$ to be from the top side to the first layer on the side and $h+d$ to be from the last layer on the side to the bottom, and the shift-amount index returns to 0 after every full rotation around the $h \times d$ perimeter, then if a shift-amount index $x$ is between two rings, the  shift-amount index $2(h + d) - x + d$ has the same shift amount.
\end{proof}

\subsubsection{Commentary on the Shift Equivalence Theorem}

The previous theorem really constrains the search space for striped nets in a way that is easy to check. When we try to get a net that covers two striped cuboids, the constraints add up. We can mitigate the constraints a little by using a whole-numbered perimeter factor $p_f$. That way, the full rotation constraint will only need to be applied once because the constraint that it repeats every $2(h_2+d_2)$ layers is already captured by the constraint from the other cuboid requiring that it repeat every $2(h_1+d_1)$ layers. (\emph{i.e.} $p_f =(2(h_2+d))/(2(h_1+d))=full\_rotation\_2/full\_rotation\_1$).

For this reason, we decided to search using whole-numbered perimeter factors. We have not fully explored what happens when we have a fractional perimeter factor greater than 1. We can only say that we did not have any success with using non-whole-numbered perimeter factors.

\begin{theorem}\label{opposite_sides}\textbf{(Opposite Sides Theorem)}\\
Stacked nets imply that the top and bottom sides of the layering cuboid will land on the opposite sides of the striped cuboid.
\end{theorem}
\begin{proof}
For a contradiction, assume without loss of generality that both the top and bottom tiles of the layering cuboid (or the $1 \times 1$ tiles of the $N\times 1\times 1$ cuboid) are at the bottom of the striped net. That means that the top of the striped cuboid would not be able to close because the striped rings would need to have more tiles than the first non-striped ring layer, but without the bottom or top of the layering cuboid, the number of tiles above the striped ring would be equal to the number of tiles on striped ring. This is a contradiction.
\end{proof}

\begin{theorem}\label{odd_perimeter_factor}\textbf{(Odd Perimeter Factor Theorem)}
If the perimeter factor between two striped cuboids is an integer, then it must be odd.
\end{theorem}
\begin{proof}
For a contradiction, let the perimeter factor be even. Consider two striped cuboids $c_1$ and $c_2$ of heights $h_1$ and $h_2$ where $h_1 > h_2$. By the Theorem \ref{opposite_sides}, the cuboids will have the bottom $1\times 1$ side at the bottom of the striped cuboid and the top $1\times 1$ side at the top of the striped cuboid. Therefore, the number of times the layering cuboid goes half way around the $h\times d$ sides of the cuboids must be an odd number. Also note that the number of times the layering cuboid goes half way around $c_2$ will be the perimeter factor times the layering cuboid goes half way $c_1$ by Corrolary \ref{perimeter_mult}. Because an even perimeter factor times an odd number gives an even number, the number of times the layering cuboid goes half way around the $h_2\times d$ side will be an even number, and the top and bottom $1\times 1$ sides of layering cuboid will be on the same side of the striped cuboid $c_2$. But by Theorem \ref{opposite_sides}, we know that this is not allowed. This is a contradiction.
Therefore, the perimeter factor must be an odd number.
\end{proof}

\subsection{Quick Note on Reference Materials and Code}\label{zenodo_appendix}

We uploaded reference images, raw output files and code to zenodo.org. The main supporting material could be found in "hidny/PolyominoNetsPaperMaterials: v1.0.0"~\cite{michael_tardibuono_2026_material}. When there is detailed results to share in the paper, the paper will default to referring to the results found here. We would also like to highlight that the main directory of that upload has a file called 'combined\_results.txt', and that file is a text file that contains most of the results mentioned in this paper.

The code we used to get the stacked net and simple-phase net results can be found in "hidny/CuboidSimplePhaseNetSearch: Release For Sharing"~\cite{michael_tardibuono_2026_stack_net}, and the code that did the thorough search can be found in two different uploads. The fast version is in "hidny/CuboidThorough: Release for sharing"~\cite{michael_tardibuono_2026_solver_thorough} and the cleaned, but slow version, is in "hidny/CuboidNetSolver: Release For Sharing"~\cite{michael_tardibuono_2026_NetSolverClean}.

\end{document}